\documentclass [11pt]{article}
\usepackage{color}
\usepackage{amsmath}
\usepackage{amssymb}
\usepackage{amscd}
\usepackage{amsthm}
\usepackage{amsopn}
\usepackage{verbatim}
\usepackage{amsmath}
\usepackage{amsfonts}
\usepackage[active]{srcltx}
\usepackage{pdfsync}
\usepackage{empheq}
\usepackage{graphicx}
\definecolor{SeaGreen}{RGB}{46,139,87}
\definecolor{Navy}{RGB}{0,0,128}
\definecolor{Maroon}{RGB}{128,0,0}
\newtheorem{theorem}{Theorem}
\newtheorem{lemma}{Lemma}[section]

\newtheorem{corollary}{Corollary}[section]

\newtheorem*{acknowledgment*}{Acknowledgment}
\newtheorem{remark}{Remark}[section]

\newcommand{\Eg}{\mathcal E} 

\newcommand{\R}{\mathbb R}
\newcommand{\C}{\mathbb C}
\newcommand{\Z}{\mathbb Z}
\newcommand{\N}{\mathbb N}

\newcommand{\CC}{\mathcal C}

\newcommand{\Union}{\mathop{\bigcup}\limits}

\def\XXint#1#2#3{{\setbox0=\hbox{$#1{#2#3}{\int}$ }
\vcenter{\hbox{$#2#3$ }}\kern-.6\wd0}}

\def\Jg {{\mathcal J}}

\newcommand{\LL}{\mathcal L}
\newcommand{\bs}{\boldsymbol}

\newcommand{\OO}{\mathcal O}

\newcommand{\QQ}{\mathcal Q}

\def\Sg {{\mathcal S}}

\def\Vg {{\mathcal V}}

\def\Tg {{\mathcal T}}
\def\Mg {{\mathcal M}}

\title{Minimization of the discrete interaction energy with smooth potentials}
  \author{ Y. Almog$^*$, Department of
  Mathematics, \\ Ort Braude College, \\ 
    Karmiel 2161002, Israel}
\date{}

\begin{document}
\maketitle

\bibliographystyle{siam}
\begin{abstract}
  We study the pair interaction on flat tori of functions whose
  Fourier coefficients are positive and decay sufficiently rapidly. In
  dimension one we find that the minimizer, up to translation, is the
  equidistant point set. In dimension two, minimizing with respect to
  triplets we find that the minimizer is the triangular lattice.
\end{abstract}
\section{Introduction}
\label{sec:1}

Let $f:\R^d \to \R$ be even and periodic in all coordinates with a
period vector $(l_1,\ldots,l_d)$. Let $\LL={\rm diag}(l_1,\ldots,l_d)$. For
instance we may consider
\begin{subequations}
  \label{eq:1}
  \begin{equation}
  f(x) = f(|x|_p) \,,
\end{equation}
where $x\in\R$ and 
\begin{equation}
|x|_p = \min_{{\bs i}\in\Z^d} |x-\LL \bs i| \,.
\end{equation}
\end{subequations}
Define the Fourier coefficients of $f$, $\{a_{\bs n}\}_{\bs n\in\Z^d}$, by
  \begin{equation}
\label{eq:2}
    a_{\bs n}= \frac{1}{|\QQ|}\int_{\QQ} e^{-i\pi\LL^{-1} {\bs n}\cdot{\bs x}}f({\bs x})\,d{\bs x}\,,
  \end{equation}
  where $\QQ=\Pi_{i=1}^d[-l_i,l_i]$.  Consider then the functional
\begin{equation}
  \label{eq:3}
\Jg(\mu) = \int_{\QQ^2} f(x-y) d\mu_x\, d\mu_y \,,
\end{equation}
where $\mu$ is a positive Borel measure for which
$\mu(\QQ)=1$. We look for the infimum of $\Jg$
\begin{displaymath}
  \Jg_m = \inf_{\mu\in\Mg(\QQ)}\Jg(\mu)\,,
\end{displaymath}
where $\Mg(\QQ)$ denotes the set of all probability
measures on $\QQ$.

For example, in one-dimension, for $f\in L^1(-l,l)$ satisfying
(\ref{eq:1}) we can write
\begin{displaymath}
  \Jg(\mu) = \int_{[-l,l]^2} f(|x-y|_p) d\mu_x\, d\mu_y \,,
\end{displaymath}
where
\begin{displaymath}
  |x|_p = \min (|x|,|2l-|x|)\,.
\end{displaymath}
The minimum is then defined by
\begin{displaymath}
  \Jg_m = \inf_{\mu\in\Mg([-l,l])}\Jg(\mu)\,,
\end{displaymath}
and the Fourier coefficients by
 \begin{displaymath}
    a_n= \int_{[-l,l]} e^{-in\pi x/l}f(|x)\,dx\,.
 \end{displaymath}

We now make the following claim whose proof is given also in
\cite{bu14} for a more general setting
\begin{lemma}
  Suppose that the Fourier coefficients of $f(x)$, defined by
  \eqref{eq:2}, are all non-negative.  Let $\mu_L$ denote the Lebesgue
  measure on $\QQ$ normalized by $|\QQ|$. Then,
  \begin{displaymath}
    \Jg_m = \Jg(\mu_L)\,.
  \end{displaymath}
Furthermore, for any $\mu\in\Mg(\QQ)$ we have
\begin{equation}
  \label{eq:4}
\Jg(\mu)-\Jg(\mu_L) = |\QQ|^2\sum_{\bs n\in\Z^d\setminus\{\bs 0\} }
  a_{\bs n}|b_{\bs n}|^2\,,
\end{equation}
where $\{b_{\bs n}\}_{\bs n\in\Z^d}$ are the Fourier
coefficients associated with $\mu$.
\end{lemma}
\begin{proof}
Let $\nu=\mu-\mu_L$. We first write
\begin{displaymath}
  \Jg(\mu)-\Jg(\mu_L) = 2\int_{\QQ^2} f(x-y) d\bs x \, d\nu_y +
  \int_{\QQ^2} f(x-y) d\nu_x\, d\nu_y \,.
\end{displaymath}
Since
\begin{displaymath}
  \int_{\QQ} f(x-y) d\bs x
\end{displaymath}
is independent of $y$, and since $\nu(\QQ)=0$, we obtain that 
\begin{equation}
  \label{eq:5}
  \Jg(\mu)-\Jg(\mu_L) = \int_{\QQ^2} f(x-y) d\nu_x\, d\nu_y \,.
\end{equation}

We next note that
  \begin{displaymath}
    f(x-y)= \sum_{\bs n\in\Z^d} a_{\bs n}e^{i\pi\LL^{-1}\bs n\cdot(\bs{x-y})} \,,
  \end{displaymath}
and
\begin{displaymath}
   d\nu_x = \sum_{\bs n\in\Z^d\setminus\{\bs 0\} } b_{\bs n}e^{i\pi\LL^{-1}\bs n\cdot \bs x} \,dx\,.
\end{displaymath}
Substituting the above into \eqref{eq:5} immediately yields
\eqref{eq:4}. 
\end{proof}

Consider then the case where $\QQ=[0,1]^N$ and
\begin{equation}
\label{eq:6}
  \mu=\frac{1}{N}\sum_{k=0}^{N-1} \delta_{\bs x_k} \,,
\end{equation}
in which $\bs x_k\in\QQ$. In this case we have for any ${\bs n}\in\Z^d$
\begin{equation}
\label{eq:7}
  b_{\bs n} = \frac{1}{N}\sum_{k=1}^N \exp
  \{i2\pi \bs n\cdot {\bs x}_k\} \,.
\end{equation}
We attempt to minimize $\Jg$ over the set $\Mg_N$ (which is isomorphic
to $[0,1)^N$) of measures given by \eqref{eq:6} for a given value of
$N$, or equivalently, over all sequences $\{x_k\}_{k=0}^{N-1}\subset\QQ$.
We can write $\Jg(\mu)$ in view of \eqref{eq:3} as
\begin{equation}
\label{eq:8}
  \Jg(\mu) = \frac{1}{N^2}\sum_{m,k =0}^{N-1} f(\bs x_m-\bs x_k) =
  \frac{(N-1)}{N}E_f + \frac{f(\bs 0)}{N} \,,
\end{equation}
where
\begin{equation}
\label{eq:9}
  E_f = \frac{1}{N(N-1)}\sum_{
    \begin{subarray}
      m,k =0 \\
      m\neq k
    \end{subarray}}^{N-1} f(\bs x_m-\bs x_k) \,.
\end{equation}
Consequently, the minimizer of $\Jg$ over measures defined by
\eqref{eq:6} is also the minimizer of $E$.

Let
\begin{displaymath}
 \Z_K = \Z  \mod K = \{0,1,\ldots,K-1\}\,,
\end{displaymath}
and ${\bs n}=(n_1,\ldots,n_d)\in \Z_K^d$. 
We proceed with the most straightforward upper bound, which is
obtained by using the
square  lattice in $\QQ$
  \begin{displaymath}
    \bs x_{\bs n} = \frac{1}{K}\bs n 
  \end{displaymath}

It can be easily verified that for every $\bs m\in\Z^d$
\begin{displaymath}
  \bar{b}_{\bs m} = \frac{(-1)^{\sum_{i=1}^dm_i}}{K^d}\sum_{\bs n\in\Z_K^d} \exp
  \Big\{i\frac{2\pi}{K} \bs m\cdot \bs n\Big\} \,.
\end{displaymath}
From which we obtain
\begin{displaymath}
  \bar{b}_{\bs m} =
  \begin{cases}
    1 & \frac{1}{K}\bs m\in\Z^d \\
    0 & \text{otherwise}\,.
  \end{cases}
\end{displaymath}
 By (\ref{eq:4}) we now easily 
\begin{equation}
\label{eq:10}
  \Jg_m -\Jg(\mu_L) \leq \sum_{\bs n\in\Z^d\setminus(0,0)} a_{K\bs n} \,.
\end{equation}
We focus in this work on the case where $a_{\bs n}$ is positive and
decays sufficiently fast for $|{\bs n}|\to\infty$. Using \eqref{eq:10} we may conclude in
such a case that
\begin{displaymath}
  a_{\bs n}|b_{\bf n}|^2\leq Jg_m -\Jg(\mu_L) \,,
\end{displaymath}
and hence that
\begin{equation}
\label{eq:11}
  |b_{\bf m}|^2\leq \sum_{\bs n\in\Z^d\setminus(0,0)} \frac{a_{K\bs n}}{ a_{\bs m}} \,.
\end{equation}
As $K\to\infty$, for rapidly decaying $\{a_n\}_{n=1}^\infty$, we may therefore
obtain a lot of information on the minimizing measure whose Fourier
coefficients $\{b_n\}_{n=1}^\infty$must also rapidly vanish in that limit.

The above sketch of ideas lies in the basis of the proofs of the
following pair of results. The first of them obtains convergence in
one dimension to an equidistant configuration
\begin{theorem}
\label{thm:1d}
  Let $f$ denote a 1-periodic function whose Fourier coefficients are
  positive and satisfy 
\begin{displaymath}
\lim_{N\to\infty}N^2 \frac{\Jg_m -\Jg(\mu_L) }{a_{N/2}} =0 \,,
  \end{displaymath} 
and that there exists $C_0>0$ such that for all $k\geq1$
\begin{equation}
\label{eq:12}
\max_{1\leq m\leq k} \frac{1}{m^2a_m} \leq \frac{C_0}{k^2a_k} \,.
\end{equation}
Let $\Sg_N$
  denote a minimizing set of \eqref{eq:9}, and let $\Tg_N$ denote the
  equidistant measure
 \begin{equation}
\label{eq:13}
    \Tg_N = \Big\{\frac{(2k-N-1)}{2N}\Big\}_{k=1}^N \,.
  \end{equation}
Then, up to arbitrary translation, $\Sg_N=\Tg_N$.
\end{theorem}

In dimension 2 the expected minimizer is the triangular lattice,
i.e., for $N=2L^2$ in the rectangle $R=[0,\sqrt{3}]\times[0,1]$ it
represented by $\Union_{k=1}^L\Union_{j=1}^{2L}( x_{kj}, y_{kj})$
  \begin{equation}
\label{eq:14}
    (x_{kj},y_{kj})=
    \begin{cases}
      \Big(\sqrt{3}\frac{k-1}{L},\frac{j-1}{2L}\Big) & j \text{ odd }\\
   \Big({\sqrt{3}\frac{k-\frac{1}{2}}L},\frac{j-1}{2L}\Big) &  j \text{ even }
    \end{cases}
  \end{equation}
Unfortunately, since estimates similar to \eqref{eq:11} cannot provide
bounds for enough Fourier coefficients, we can only consider minimization with
respect to triplets, i.e., we assume that $\{(x_n,y_n)\}_{n=1}^{2L^2}$
satisfies
\begin{equation}
\label{eq:15}
\begin{cases}
   (x_{n+2L^2/3},y_{n+2L^2/3})= \Big(x_n+\frac{\sqrt{3}}{2L},y_n+\frac{1}{2L}\Big) &
   n\in\Big[1,\frac{2L^2}{3}\Big]\cap\Z \\
  (x_{n+4L^2/3},y_{n+4L^2/3})= \Big(x_n+\frac{\sqrt{3}}{L},y_n\Big)
\end{cases}
\end{equation}
The triangular lattice \eqref{eq:14} is given in that case by $\Tg_L=
\Union_{k=1}^{2L/3}\Union_{j=1}^L( x_{kj}, y_{kj})$ where
  \begin{displaymath}
  x_{kj}=\frac{3(2k-1)}{4L}\quad ; \quad y_{kj}=\frac{4j+[1+(-1)^k]}{4L}\,,
\end{displaymath}
where each $( x_{kj}, y_{kj})$ is accompanied by additional pair of
points in the lattice according to \eqref{eq:15}. 

We can now state the main result in the two-dimensional case
\begin{theorem}
\label{thm:2d}
  Let $f$ denote a 2-periodic function whose Fourier coefficients
  $\{a_{nm}\}_{(n,m)\in\Z^2}$ are
  positive. Let further 
\begin{displaymath}
  \tilde{\epsilon}_L = \frac{L^2}{2}\sum_{(p,q)\in\Z^2\setminus\{(0,0)\}} (p^2+q^2)a_{pL,qL}[1+(-1)^{p+q}] \,.
\end{displaymath}
 and suppose that
 \begin{equation}
\label{eq:16}
\lim_{L\to\infty} e^{\lambda L}\frac{\tilde{\epsilon}_L}{a_{2L/3,L}} =0 \,,
  \end{equation}
and that there exists $C_0>0$ and $L_0>0$ such that for all
$L>L_0$
\begin{equation}
\label{eq:17}
\max_{
  \begin{subarray}{c}
    |m|\leq \frac{2L}{3} \\
    |n|\leq |L|
  \end{subarray}}
a_{mn}\geq C_0a_{2L/3,L}\,.
\end{equation}
  
  Let $\Sg_L$ denote a minimizing set of \eqref{eq:9} satisfying
  \eqref{eq:15}, and let $\Tg_L$ be given by \eqref{eq:63}. Then,
  for sufficiently large $L$, it holds that $\Sg_L=\Tg_L$.
\end{theorem}

The minimization problem \eqref{eq:9} has been frequently addressed in
the literature. In \cite{go03} it has been established that for
$1-$periodic $f:\R\to\R$, which is convex and monotone decreasing in
$[0,1]$, the unique minimizer (up to translation) is the equidistant
sequence \eqref{eq:13}. Note that the Fourier expansion of any real, even,
strictly convex function in $[-1/2,1/2]$ consists of real positive
coefficients only. 

In dimension 2 consider the functional
\begin{displaymath}
\Eg(\CC):= \liminf_{R \to \infty} \frac{1}{|\CC \cap B_R|} \sum_{x, y \in \CC \cap B_R, x\neq y}
f(|x-y|)
\end{displaymath}
 where $\CC$ is a non-empty discrete point configuration of $\R^2$ of
 uniform density and $f:\R_+\to \R$ is  a completely monotone function
 of the  squared distance, i.e. when $f(r)= g(r^2)$
with $g$ a smooth completely monotone function on $\R_+$
i.e. satisfying $(-1)^k g^{(k)}(r) \geq 0$  for all $r \geq 0$ for every
integer $k\geq 0$. In particular the Gaussian $f(r)=e^{-tr^2}$  is  a completely monotone function
 of the  squared distance. By the Cohn-Kumar conjecture the minimizer
 of $\Eg$ is the triangular lattice (see \eqref{eq:14}). 

In \cite{pese20} it is demonstrated that
if the Cohn-Kumar conjecture is true in $\R^2$ then the triangular
lattice is the global minimizer also in the case where $f(r)=-\log
r$. The method used in \cite{pese20} include integration of the periodic heat
kernel with respect to time to obtain the kernel of the
periodic inverse Laplacian. We shall refer to \cite{pese20} again in
\S~\ref{sec:3}. 

Beyond the above there is a long list of works that addresses
\eqref{eq:9} in various settings. To name just a few we mention
\cite{saff2013logarithmic,hasa04,beetal21,coetal22}. The problem is
related to the well-known Abrikosov problem in superconductivity via
the Mellin transform \cite{ab57,mo88}.   

In the next section we consider the one-dimensional problem and prove
Theorem \ref{thm:1d}. In Section 3 we consider the two-dimensional
problem and prove Theorem \ref{thm:2d}.

\section{1D Minimization}
\label{sec:2}

Let
\begin{displaymath}
  \epsilon_N = \sum_{n\in\Z^\setminus\{0\}} a_{nN} \,.
\end{displaymath}
Let $\{x_k\}_{k=1}^N$ denote the minimizer of (\ref{eq:8}) for
  some $f\in L^2(-1/2,1/2)$. Without loss of generality we may assume that
  $x_1\leq x_2\leq\ldots\leq x_N$. We now show
\begin{lemma}
Suppose that 
\begin{equation}
    \label{eq:18}
  \lim_{N\to\infty}N \frac{Jg_m -\Jg(\mu_L) }{a_{N/2}} \leq \lim_{N\to\infty}N \frac{\epsilon_N}{a_{N/2}}=0 \,,
\end{equation}
and that there exists $C_0>0$ such that for all $k\geq1$ \eqref{eq:12} is
satisfied. Then, there exists $C>0$ such that for sufficiently large $N$,
  \begin{equation}
    \label{eq:19}
\min_{\alpha\in[0,1]} \max_{1\leq k\leq N}|x_k-(-1/2+2k/N+\alpha)|\leq C\frac{[\epsilon_N]^{1/2}}{[Na_{N/2}]^{1/2}}\,.
  \end{equation}
\end{lemma}
\begin{proof}
We first note that for any $p\leq2$ we have by (\ref{eq:12})
\begin{equation}
\label{eq:20}
  \max_{1\leq m\leq k} \frac{1}{m^pa_m} \leq k^{2-p}\max_{1\leq m\leq k}
  \frac{1}{m^2a_m}\leq \frac{C_0}{k^pa_k}\,.
\end{equation}

 Let $b_n$ be given by \eqref{eq:7}. By (\ref{eq:4}) we have that
  \begin{displaymath}
    |b_n| \leq \Big[\frac{\epsilon_N}{a_n}\Big]^{1/2} \,.
  \end{displaymath}
Let $z_k=e^{i2\pi x_k}$. Without loss of generality we assume that
\begin{equation}
\label{eq:21}
  \prod_{i=1}^N z_i =1 \,,
\end{equation}
otherwise we translate $\{x_k\}_{k=1}^N$ by the necessary amount. Let
further $e_0=1$ and
\begin{equation}
\label{eq:22}
  e_k(z_1,\ldots,z_N) = \sum_{1\leq i_1<i_2<\ldots<i_k\leq N}  z_{i_1}\cdots  z_{i_k} \,.
\end{equation}
By Newton's identities we have
\begin{displaymath}
  ke_k = N \sum_{m=1}^k (-1)^{m-1}e_{k-m}b_m\,.
\end{displaymath}
For $k=1$ we thus obtain
\begin{displaymath}
  e_1 = Nb_1
\end{displaymath}
Invoking inductive arguments we suppose that for any $1\leq k\leq K-1$,
where \break $2\leq K\leq N/2-1$
\begin{equation}
\label{eq:23}
\Big|e_k-(-1)^{k-1}\frac{N}{k}b_k\Big|\leq2C_0\frac{N^2}{k^2}\frac{\epsilon_N}{a_k}\,.
\end{equation}
Then, 
\begin{displaymath}
  \Big|e_K-(-1)^{K-1}\frac{N}{K}b_K\Big|\leq N^2 \sum_{m=1}^{K-1}\frac{1}{m(K-m)} |b_{K-m}|\,|b_m|
  +2C_0N^2 \epsilon_N\sum_{m=1}^{K-1} \frac{|b_m|}{m}\frac{1}{(K-m)^2a_{K-m}} \,.
\end{displaymath}
For the first sum on the right-hand side we have, by (\ref{eq:12})
\begin{displaymath}
  N^2\sum_{m=1}^{K-1}\frac{1}{m(K-m)} |b_{K-m}|\,|b_m| \leq N^2\sum_{m=1}^{K-1}\frac{1}{a_mm^2} a_m|b_m|^2\leq
N^2\max_{1\leq m\leq K}\frac{\epsilon_N}{m^2a_m}\leq
\frac{C_0N^2}{K^2}\frac{\epsilon_N}{a_K}\,, 
\end{displaymath}
whereas for the second term we have, using \eqref{eq:12} once again,
\begin{multline*}
  N^2 \epsilon_N\sum_{m=1}^{K-1}
  \frac{|b_m|}{m}\frac{1}{(K-m)^2a_{K-m}} \leq C_0\frac{N^2 \epsilon_N}{K^2a_K}\sum_{m=1}^{K-1}
  \frac{|b_m|}{m} \\\leq C_0\frac{N^2 \epsilon_N}{K^2a_K}\Big[\sum_{m=1}^{K-1}\frac{1}{a_mm^2} \Big]^{1/2}
\Big[\sum_{m=1}^{K-1} a_m|b_m|^2\Big]^{1/2}\leq 2C_0^{3/2}\frac{N^2}{K^{5/2}} \frac{\epsilon_N^{3/2}}{a_K^{3/2}} \,.
\end{multline*}
Combining the above then yields
\begin{displaymath}
\Big|e_K-\frac{N}{K}b_K\Big|\leq C_0\frac{N^2}{K^2}\frac{\epsilon_N}{a_K}
 \bigg[1+2C_0^{1/2}K^{-1/2}\frac{\epsilon_N^{1/2}}{a_K^{1/2}}\bigg] \,.   
\end{displaymath}
From (\ref{eq:18}) we now get (\ref{eq:23}) for sufficiently large $N$.
For $N/2+1\leq k\leq N$ we have by (\ref{eq:21})
\begin{displaymath}
  e_k = \bar{e}_{N-k}  \,.
\end{displaymath}
Hence, for all $N/2+1\leq k\leq N$
\begin{equation}
\label{eq:24}
  \Big|e_k-\frac{N}{k}\bar{b}_{N-k}\Big|\leq2C_0\frac{N^2\epsilon_N}{k^2a_k}
\end{equation}

We now attempt to use the identity 
\begin{equation}
\label{eq:25}
  P_N(t) = \prod_{i=1}^N (t-z_i) = \sum_{k=0}^N(-1)^ke_kt^{N-k} \,,
\end{equation}
to obtain an approximation for $\{z_i\}_{i=1}^N$. 
Let $\alpha_p=e^{i2\pi p/N}$. Let further $S_\rho=\partial B(\alpha_p,\rho)$ for some $0<\rho\ll1/N$
as $N\to\infty$. On $S_\rho$ we have 
\begin{equation}
\label{eq:26}
  |t^N-1 -N\alpha_p^{N-1}(t-\alpha_p)|\leq CN^2\rho^2 \,.
\end{equation}
From (\ref{eq:23}) and (\ref{eq:24}) we obtain that, for sufficiently
large $N$, using the fact that $|t^k|\leq2$ on $S_\rho$ for all $1\leq k\leq N-1$,
\begin{displaymath}
  \max_{t\in S_\rho}\Big|\sum_{k=1}^{N-1}(-1)^ke_kt^{N-k}\Big| \leq
  4N\sum_{k=1}^{N/2}\frac{|b_k|}{k} +
  4C_0N^2\sum_{k=1}^{N/2}\frac{\epsilon_N}{k^2a_k}
\end{displaymath}
For the first term on the right-hand-side we have by (\ref{eq:12})
\begin{displaymath}
  4N\sum_{k=1}^{N/2}\frac{|b_k|}{k}\leq 4N
  \Big[\sum_{k=1}^{N/2}\frac{1}{k^2a_k}\Big]^{1/2} \Big[\sum_{k=1}^{N/2}
  a_k|b_k|^2\Big]^{1/2}\leq
  4C_0N^{1/2}\frac{\epsilon_N^{1/2}}{a_{N/2}^{1/2}} \,.
\end{displaymath}
For the second term we conclude from (\ref{eq:12}) that there exists
$C>0$ such that for sufficiently large $N$
\begin{displaymath}
  4C_0N^2\sum_{k=1}^{N/2}\frac{\epsilon_N}{k^2a_k}\leq C \frac{\epsilon_N}{a_{N/2}} \,.
\end{displaymath}
Consequently, by (\ref{eq:18})
\begin{displaymath}
  \max_{t\in S_\rho}\Big|\sum_{k=1}^{N-1}(-1)^ke_kt^{N-k}\Big| \leq C\frac{[N\epsilon_N]^{1/2}}{a_{N/2}^{1/2}}\,.
\end{displaymath}

We thus conclude that whenever 
\begin{displaymath}
  \frac{\epsilon_N^{1/2}}{[Na_{N/2}]^{1/2}} <  \frac{\rho}{2} \ll    \frac{1}{N}\,,
\end{displaymath}
we have by \eqref{eq:26} that
\begin{displaymath}
   \min_{t\in S_\rho}|t^N-1| \geq \frac{N\rho}{2}>\max_{t\in S_\rho}\Big|\sum_{k=1}^{N-1}(-1)^ke_kt^{N-k}\Big| \,.
\end{displaymath}
By Rouch\'e's theorem and \eqref{eq:25}, $P_N(t)$ must have a unique simple zero inside
$S_\rho$, which readily proves (\ref{eq:19}).
\end{proof}

We recall \eqref{eq:21} which readily yields 
\begin{equation}
  \label{eq:27}
\sum_{n=1}^N x_n=0 \,,
\end{equation}
and let $\Tg_N$ be given by \eqref{eq:13}.  
Let $\Sg_N=\{x_n\}_{n=1}^N$ denote the minimizer of $\Jg$ satisfying
\eqref{eq:27}. Then, set ${\bs \delta}_N=\Sg_N-\Tg_N$. Let further
\begin{displaymath}
  {\bs V}_n = \nabla b_n(\Tg_N)= i\frac{2\pi n}{N}
  \begin{bmatrix}
    e^{i\frac{\pi n}{N}} \\
    e^{i\frac{3\pi n}{N}} \\
    \vdots  \\
     e^{i\frac{(2N-1)\pi n}{N}}
  \end{bmatrix}\,.
\end{displaymath}
We can then state
\begin{lemma}
Suppose that $\{x_n\}_{n=1}^N\in[-1/2,1/2]^N$ and $f\in L^2(-1/2,1/2)$
satisfy \eqref{eq:12} and
\begin{equation}
  \label{eq:28}
\lim_{N\to\infty}N^2 \frac{\epsilon_N}{a_{N/2}}=0\,.
\end{equation}
(Note that \eqref{eq:28} is a stronger restriction than \eqref{eq:18}.)
 Then, there exist $C>0$ and $N_0>0$, such that for all $N>N_0$
 \begin{equation}
   \label{eq:29}
\sum_{n=1}^{\lfloor N/2\rfloor}\frac{1}{n^2}|b_n|^2 \geq \frac{C}{N}\|{\bs \delta}_N\|_2^2 \,,
 \end{equation}
where $\lfloor\cdot\rfloor:\R\to\Z$ denotes the integer part. 
\end{lemma}
\begin{proof}
{\em Step 1:}   Show that there exists $C>0$ such that 
  \begin{equation}
    \label{eq:30}
\sum_{n=1}^{\lfloor N/2\rfloor} \frac{1}{n^2}|{\bs V}_n \cdot{\bs \delta}_N|^2\geq \frac{2\pi^2}{N}\|{\bs \delta}_N\|_2^2 \,.
  \end{equation}

Since $\{V_n\}_{n=1}^N$ spans $\C^N$, we may write, in view of
\eqref{eq:27},
\begin{displaymath}
  {\bs \delta}_N=\sum_{n=1}^{N-1}\frac{N^{1/2}}{2\pi n}C_n{\bs V}_n\,,
\end{displaymath}
where $\{C_n\}_{n=1}^{N-1}\in\C^{N-1}$ satisfies, since ${\bs \delta}_N\in\R^N$
\begin{displaymath}
  C_n=\bar{C}_{N-n} \quad \forall1\leq n\leq N/2 \,.
\end{displaymath}
From the above it can be easily verified that
\begin{displaymath}
  \|{\bs \delta}_N\|_2^2=\sum_{n=1}^{N-1}|C_n|^2 \leq 2\sum_{n=1}^{\lfloor N/2\rfloor}|C_n|^2\,.
\end{displaymath}
The proof of \eqref{eq:30} can now be easily completed by using the
above and the fact that 
\begin{displaymath}
  \sum_{n=1}^{\lfloor N/2\rfloor} \frac{1}{n^2}|{\bs V}_n \cdot{\bs \delta}_N|^2=\frac{4\pi^2}{N}\sum_{n=1}^{\lfloor N/2\rfloor}|C_n|^2\,.
\end{displaymath}

{\em Step 2:} Prove that there exists $C>0$ such that
\begin{equation}
\label{eq:31}
  \sum_{n=1}^{\lfloor N/2\rfloor}\frac{1}{n^2}|b_n-{\bs V}_n \cdot{\bs \delta}_N|^2 \leq
  CN^2\|{\bs \delta}_N\|^2_\infty\|{\bs \delta}_N\|_2^2\,.
\end{equation}

We write the Taylor expansion of $b_n$, for $1\leq n\leq N/2$, around
$\Tg_N$
\begin{displaymath}
b_n =   {\bs V}_n \cdot{\bs \delta}_N -
\frac{2\pi^2n^2}{N}\sum_{k=1}^Ne^{i2\pi n\tilde{x}_k}|\delta_N^k|^2 \,,
\end{displaymath}
where
\begin{displaymath}
  \tilde{x}_k= \frac{2k-1}{2N}+ \theta\Big(x_k- \frac{2k-1}{2N}\Big)\,.
\end{displaymath}
Consequently,
\begin{displaymath}
  \frac{1}{n}|b_n-{\bs V}_n \cdot{\bs \delta}_N|\leq C\|{\bs \delta}_N\|^2_2\leq CN^{1/2}\|{\bs
    \delta}_N\|_\infty \,\|{\bs \delta}_N\|_2\,,
\end{displaymath}
and hence
\begin{equation}
\label{eq:32}
   \sum_{n=1}^{\lfloor N/2\rfloor}\frac{1}{n^2}|b_n-{\bs V}_n \cdot{\bs \delta}_N|^2 \leq
   CN^2 \|{\bs  \delta}_N\|_\infty^2 \,\|{\bs \delta}_N\|_2^2\,.
\end{equation}
Substituting the above  into \eqref{eq:32}
yields \eqref{eq:31}.

{\em Step 3:} Prove \eqref{eq:30}.

As
\begin{displaymath}
  \sum_{n=1}^{\lfloor N/2\rfloor}\frac{|b_n|^2}{n^2} \geq \frac{1}{2}
  \sum_{n=1}^{\lfloor N/2\rfloor}\frac{|{\bs V}_n \cdot{\bs \delta}_N|^2}{n^2} -
  \sum_{n=1}^{\lfloor N/2\rfloor}\frac{|b_n-{\bs V}_n \cdot{\bs \delta}_N |^2}{n^2}  \,,
\end{displaymath}
we obtain from \eqref{eq:30} and \eqref{eq:31} that
\begin{equation}
\label{eq:33}
  \sum_{n=1}^{\lfloor N/2\rfloor}\frac{|b_n|^2}{n^2} \geq \frac{C}{N}\Big(\|{\bs \delta}_N\|_2^2-N^3\|{\bs \delta}_N\|^2_\infty\|{\bs \delta}_N\|_2^2)
\end{equation}
We can now conclude \eqref{eq:29} from \eqref{eq:33}, \eqref{eq:28}
and  \eqref{eq:19}. 
\end{proof}

We now establish that a minimizing set consists of $N$ equidistant points.
\begin{proof}[Proof of Theorem \ref{thm:1d}]
  Let $b_n=b_n(\Sg_N)$. We first observe that
  \begin{equation}
\label{eq:34}
    \Jg(\Sg_N)-\Jg(\Tg_N)\geq \sum_{n=1}^{\lfloor N/2\rfloor} a_n|b_n|^2-\sum_{n=1}^\infty a_{nN}(1-|b_{nN}|^2)\,.
  \end{equation}
We first obtain a lower bound for the first term on the
right-hand-side using \eqref{eq:29}
\begin{equation}
\label{eq:35}
   \sum_{n=1}^{\lfloor N/2\rfloor} a_n|b_n|^2 \geq CN^2a_{N/2}\sum_{n=1}^{\lfloor N/2\rfloor}
   \frac{|b_n|^2}{n^2}\geq CNa_{N/2} \|{\bs \delta}_N\|_2^2\,. 
\end{equation}

We next obtain an upper bound for the second term on the
right-hand-side of \eqref{eq:35}. Since $\Tg_N$ is a maximum point for
$|b_{nN}|^2$ for all $n\in\N$ where $b_{nN}(\Sg_N)=1$, we obtain that
\begin{equation}
\label{eq:36}
  1-|b_{nN}|^2 = \frac{1}{2}{\bs \delta}_N\cdot\tilde{\Mg}_{nN}{\bs \delta}_N\,,
\end{equation}
where the Hessian matrix $\tilde{\Mg}_{nN}$ is given by
\begin{displaymath}
  \tilde{\Mg}_{nN} = D^2(|b_{nN}|^2)(\tilde{\Sg}_N)\,,
\end{displaymath}
in which
\begin{displaymath}
  \tilde{\Sg}_N=\Tg_N+\theta(\Sg_N-\Tg_N)\,.
\end{displaymath}
Set
\begin{displaymath}
  {\bs d}_N=
  \begin{bmatrix}
    |\delta_N^1| \\
    |\delta_N^2| \\
    \vdots \\
    |\delta_N^N|
  \end{bmatrix}
\text{ and let } \Mg_{nN}^0=8\pi^2n^2
\begin{bmatrix}
  (N+1) & 1 & 1 & \cdots & 1 \\
  1 & (N+1) & 1 & \cdots & 1 \\
  \vdots  & \vdots    & \vdots & \ldots   & \vdots \\
  1 & 1 & 1 & \ldots & (N+1) 
\end{bmatrix}
\,.
\end{displaymath}
Further, let the elements of $\Mg\in M^{N\times N}$ be given by
$\Mg_{nk}=|\tilde{\Mg}_{nk}|$. Clearly
\begin{displaymath}
  |{\bs \delta}_N\cdot\tilde{\Mg}_{nN}{\bs \delta}_N|\leq {\bs d}_N\cdot\Mg_{nN}{\bs
    d}_N\leq {\bs d}_N\cdot\Mg_{nN}^0{\bs
    d}_N\,,
\end{displaymath}
and hence, by \eqref{eq:36}
\begin{displaymath}
  1-|b_{nN}|^2\leq 16\pi^2n^2N \|{\bs \delta}_N\|_2^2 \,.
\end{displaymath}
Consequently,
\begin{equation}
\label{eq:37}
\sum_{n=1}^\infty a_{nN}(1-|b_{nN}|^2)\leq CN\|{\bs
  \delta}_N\|_2^2\sum_{n=1}^\infty n^2a_{nN} \,. 
\end{equation}
By \eqref{eq:18} it holds that for sufficiently large $N$
\begin{displaymath}
  a_N\leq\frac{a_{N/2}}{N} \,.
\end{displaymath}
Let $l=\lfloor\log_2n\rfloor+1$. From the above we may conclude that
\begin{displaymath}
  a_{nN}\leq\frac{a_{\lfloor2^{-l}nN\rfloor}}{N^l} \,.
\end{displaymath}
Combining the above with \eqref{eq:12} then yields 
\begin{displaymath}
  a_{nN}\leq C\frac{a_{N/2}}{N^l} \,.
\end{displaymath}
Consequently
\begin{displaymath}
  \sum_{n=1}^\infty n^2a_{nN} \leq Ca_{N/2}\sum_{l=1}^\infty
  \Big[\frac{8}{N}\Big]^l\leq\frac{C}{N}a_{N/2} \,,
\end{displaymath}
which when substituted into \eqref{eq:37} yields
\begin{displaymath}
  \sum_{n=1}^\infty a_{nN}(1-|b_{nN}|^2)\leq C\|{\bs
  \delta}_N\|_2^2\,a_{N/2} \,,
\end{displaymath}
Substituting the above, together with \eqref{eq:35} into \eqref{eq:34} yields, for sufficiently
large $N$,
\begin{displaymath}
  \Jg(\Sg_N)-\Jg(\Tg_N)\geq0 \,,
\end{displaymath}
where equality is achieved only when $\Sg_N=\Tg_N$.  The Theorem
is proved.
\end{proof}

\section{Two dimensions}
\label{sec:3}
Let $f\in L^2_{loc}(\R^2)$ be given by
\begin{equation}
\label{eq:38}
  f(x_1,x_2)=\sum_{(p,q)\in\Z^2} a_{pq}e(px_1/\sqrt{3}+qx_2) \,,
\end{equation}
where $a_{pq}\in\R_+$ for all $(p,q)\in\Z^2$ and
\begin{displaymath}
  e(s)=e^{i2\pi s}\,.
\end{displaymath}
Obviously, $f$ is periodic in $\R^2$ with the following rectangle as its
basic cell
\begin{equation}
\label{eq:39}
  R=[0,\sqrt{3}]
  \times[0,1]  \,.
\end{equation}
More precisely, it holds that
\begin{displaymath}
  f(x_1,x_2)=f(x_1+\sqrt{3},x_2)=f(x_1,x_2+1)\,.
\end{displaymath}
Let for some $L\in[2,\infty)\cap\Z$
\begin{displaymath}
  \epsilon_L = \frac{1}{2}\sum_{(p,q)\in\Z^2\setminus\{(0,0)\}} a_{pL,qL}[1+(-1)^{p+q}] \,.
\end{displaymath}
and 
\begin{displaymath}
  \tilde{\epsilon}_L = \frac{L^2}{2}\sum_{(p,q)\in\Z^2\setminus\{(0,0)\}} (p^2+q^2)a_{pL,qL}[1+(-1)^{p+q}] \,.
\end{displaymath}

\subsection{Upper bound}
\label{sec:upper}

We begin with the following simple observation. 
\begin{lemma}
\label{lem:lattice}
Let $N=2L^2$. Let $\{x_{km}\}_{(k,m)=(1,1)}^{(L,2L)}\in R\times R$, where $R$
is given by \eqref{eq:39}, be arranged on the following lattice, which
is given by \eqref{eq:14}. We repeat here \eqref{eq:14} for the
convenience of the reader
  \begin{displaymath}
    (x_{kj}^1,x_{kj}^2)=
    \begin{cases}
      \Big(\sqrt{3}\frac{k-1}{L},\frac{j-1}{2L}\Big) & j \text{ odd }\\
   \Big({\sqrt{3}\frac{k-\frac{1}{2}}L},\frac{j-1}{2L}\Big) &  j \text{ even }
    \end{cases}
  \end{displaymath}
 Then
 \begin{equation}
\label{eq:40}
   \Jg(\mu)-\Jg(\tilde{\mu}_L) =
   \frac{1}{2}\sum_{(p,q)\in\Z^2\setminus\{(0,0)\}}[1+(-1)]^{p+q}] a_{Lp,Lq} \,,
 \end{equation}
where
\begin{displaymath}
  \mu=\frac{1}{N}\sum_{k=1}^L\sum_{j=1}^{2L}\delta(x_1-x_{kj}^1)\delta(x_2-x_{kj}^2)\,,
\end{displaymath}
and $\tilde{\mu}_L=[3]^{-1/2}\mu_L$ ($\mu_L$ being the Lebesgue measure). 
\end{lemma}
\begin{proof}
Given the Fourier expansion \eqref{eq:38} we have, as in \eqref{eq:4}
\begin{equation}
\label{eq:41}
  \Jg(\mu)-\Jg(\tilde{\mu}_L) =\sqrt{3} \sum_{(p,q)\in\Z^2\setminus(0,0)}
  a_{pq}|b_{pq}|^2\,,
\end{equation}
where
\begin{displaymath}
   b_{pq}(\mu) = \frac{1}{2\sqrt{3}L^2}\sum_{k=1}^L\sum_{j=1}^{2L}e(\sqrt{1/3}px_{kj}^1+qx_{kj}^2) \,.
\end{displaymath}
A simple computation yields
\begin{multline*}
  b_{pq}=
  \frac{1}{2\sqrt{3}L^2}\Big[
  \sum_{k,j=0}^{L-1}e^{2\pi i\Big(p\frac{k}{L}+q\frac{j}{L}\Big)} +
  \sum_{k,j=0}^{L-1}e^{2\pi i\Big(p\frac{2k+1}{2L}+q\frac{2j+1}{2L}\Big) }
  \Big]=\\\frac{1+(-1)^{\frac{p+q}{L}}}{2\sqrt{3}}
      \begin{cases}
      1 & \Big(\frac{p}{L},\frac{q}{L}\Big)\in\Z^2 \\
      0 & \text{otherwise} \,.
    \end{cases}
\end{multline*}
Substituting the above into \eqref{eq:41} yields \eqref{eq:40}\,.
\end{proof}

Suppose now that there exists $\lambda_0>0$ such that for all $\lambda<\lambda_0$
\eqref{eq:16} and \eqref{eq:17} hold true. We repeat here \eqref{eq:16}
and \eqref{eq:17} for the convenience of the reader
  \begin{displaymath}
\lim_{L\to\infty} e^{\lambda L}\frac{\tilde{\epsilon}_L}{a_{2L/3,L}} =0 \,,
  \end{displaymath}
and there exists $C_0>0$ and $L_0>0$ such that for all
$L>L_0$
\begin{displaymath}
\max_{
  \begin{subarray}{c}
    |m|\leq \frac{2L}{3} \\
    |n|\leq |L|
  \end{subarray}}
a_{mn}\geq C_0a_{2L/3,L}
\end{displaymath}
We can now conclude that
\begin{corollary}
  Let $N=2L^2$ for some $L\in\N$. Let $\mu_m(N)$ denote the minimizing
  measure in the set \eqref{eq:6}. Suppose that $f$ satisfies
  \eqref{eq:38}, \eqref{eq:16}, and \eqref{eq:17}.  Then, there exists
  $\lambda_0>0$ such that for any $\lambda<\lambda_0$
  \begin{equation}
    \label{eq:42}
\lim_{L\to\infty}e^{\lambda L}\max_{
  \begin{subarray}
    |p|\leq2L/3 \\
    |q|\leq L
  \end{subarray}}
|b_{pq}(\mu_m)|=0 \,.
  \end{equation}
\end{corollary}
\begin{proof}
  By the definition of $\epsilon_L$ it holds that
  \begin{equation}
    \label{eq:43}
\Jg(\mu_m)-\Jg(\tilde{\mu}_L)\leq \epsilon_L
  \end{equation}
Consequently, by \eqref{eq:41} (which holds for any $\mu$ satisfying
\eqref{eq:6}),
\begin{displaymath}
  \sqrt{3}\sum_{(p,q)\in\Z^2\setminus(0,0)}
  a_{pq}|b_{pq}|^2 \leq \epsilon_L\,.
\end{displaymath}
Consequently,  whenever
$|p|^\leq2L/3$ and $|q|\leq L$ it holds that there exists $C>0$ such that
\begin{displaymath}
  |b_{pq}|\leq C\Big[\frac{\epsilon_L}{a_{pq}}\Big]^{1/2}\,.
\end{displaymath}
By \eqref{eq:16} and \eqref{eq:17} we may now conclude \eqref{eq:42}. 
\end{proof}

Let
\begin{displaymath}
  S=\{ (p,q)\in\Z^2\setminus\{(0,0)\}\,|\,\min(|p|,|q|)\leq L-1\cap\max (|p|,|q|)\leq L\} 
\end{displaymath}
We note that by \eqref{eq:42} it holds that for any $\lambda<\lambda_0$ there
exists $L_0>0$ such that for all $L>L_0$ it holds that
\begin{equation}
\label{eq:44}
  (p,q)\in S\Rightarrow|b_{pq}|\leq e^{-\lambda L} \,.
\end{equation}

Unfortunately, \eqref{eq:44} does not consist of enough independent
inequalities necessary to find the $4L^2$ unknowns
$\{(x_n,y_n)\}_{n=1}^{2L^2}$ by the technique which follows.  We shall
therefore consider here minimization of \eqref{eq:9} with respect to
triplets.

\subsection{Minimization with respect to triplets}
\label{sec:triplets}

In the following we attempt to minimize $\Jg$ over the set of
configurations $\{(x_n,y_n)\}_{n=1}^{2L^2}$ which obey \eqref{eq:15},
which we repeat here for the convenience of the reader
\begin{align*}
  & (x_{n+2L^2/3},y_{n+2L^2/3})= \Big(x_n+\frac{\sqrt{3}}{2L},y_n+\frac{1}{2L}\Big) \quad
   n\in\Big[1,\frac{2L^2}{3}\Big]\cap\Z \\
  & (x_{n+4L^2/3},y_{n+4L^2/3})= \Big(x_n+\frac{\sqrt{3}}{L},y_n\Big) 
\end{align*}

Let 
\begin{equation}
  \label{eq:45}
S_1=\{ (p,q)\in\Z^2\setminus\{(0,0)\}\,|\,\max (3|p|/2,|q|)\leq L\} 
\end{equation}
Let further
\begin{displaymath}
  z_n=e(x_n/\sqrt{3}) \quad ; \quad w_n=e(y_n) \,. 
\end{displaymath}
Under \eqref{eq:15} we have
\begin{displaymath}
  b_{mn}=\sum_{k=1}^{2L^2}z_k^mw_k^n=
  \Big[1+e\Big(\frac{m}{L}\Big)+e\Big(\frac{m+n}{2L}\Big]
  \sum_{k=1}^{\frac{2L^2}{3}}z_k^mw_k^n\,.  
\end{displaymath}
By \eqref{eq:44} there exists $\lambda_0$ such that for all $\lambda<\lambda_0$ and
sufficiently large $L$
\begin{equation}
\label{eq:46}
  \Big|\sum_{k=1}^{\frac{2L^2}{3}}z_k^mw_k^n\Big|\leq e^{-\lambda L}\,,\;
  \forall(m,n)\in S \,,
\end{equation}
as long as
\begin{equation}
  \label{eq:47}
1+e\Big(\frac{m}{L}\Big)+e\Big(\frac{m+n}{2L}\Big)\neq0
\end{equation}
taking into account that otherwise
\begin{displaymath}
\Big|1+e\Big(\frac{m}{L}\Big)+e\Big(\frac{m+n}{2L}\Big)\Big|\geq
  \frac{1}{L} \,.
\end{displaymath}
Note that \eqref{eq:47} is satisfied for all $(m,n)\in S_1$ unless
\begin{displaymath}
    \frac{m}{L}=\pm\frac{1}{3} \text{ and } \frac{m+n}{2L}=\mp\frac{1}{3}\,.
\end{displaymath}
Consequently, \eqref{eq:15} leads to the solvable system
\begin{equation}
\label{eq:48}
  \Big|\sum_{k=1}^{\frac{2L^2}{3}}z_k^mw_k^n\Big|\leq e^{-\lambda L}\,,\; \forall(m,n)\in S_1\setminus\{(\pm
  2L/3,0),(\pm L/3,\pm L)\}\,.
\end{equation}
Above and in the sequel the inequality $A_L\leq e^{-\lambda L}$ means that there
exist $\lambda_0>0$ (not necessarily the same for each inequality of this
form) such that for every $0<\lambda<\lambda_0$ it holds that
\begin{displaymath}
  \lim_{L\to\infty}e^{\lambda L}A_L=0\,.
\end{displaymath}

We continue by obtaining the solutions of \eqref{eq:48}
\begin{lemma}
Let $\{(z_k,w_k)\}_{k=1}^{\frac{2L^2}{3}}\subset[{\mathbb S}^1]^2$ satisfy
\eqref{eq:48} for some $\delta>0$ and $L>1$.
Then, there exist $\lambda_0>0$, $(\alpha,\beta)\in\R^2$ such that 
  \begin{equation}
\label{eq:49}
  \lim_{L\to\infty}e^{\lambda L}\big|(z_k^{L/3},w_k^L)-(-1)^{k-1}(e^{i\alpha},e^{i\beta})\big|=0\,,\;\forall k\in[1,L^2]\cap\Z\,,
\end{equation}
for all $0<\lambda<\lambda_0$. 
\end{lemma}
\begin{proof}
Let
  \begin{equation}
\label{eq:50}
    v_{n,m}=\sqrt{\frac{3}{2L^2}}\big[z_1^nw_1^m,z_2^nw_2^m,
    \ldots,z_{L^2}^nw_{L^2}^m\big]^T\,,\;(n,m)\in\Z^2 \,.
  \end{equation}
Let further
\begin{displaymath}
  \langle a,b\rangle=\sum_{k=1}^M a_k\bar{b}_k \,, \; a,b\in\C^M
\end{displaymath}
denote the standard inner product  in $\C^M$. Let $S_{2i}\in\Z$,
$i=1,2$ be given by 
\begin{displaymath}
  S_{21}=\{[-L/3+1,L/3]\cap\Z\}\quad ;\quad S_2=\{[-L/2+1,L/2]\cap\Z\}
\end{displaymath}

Let $\Mg\in\C^{2L^2/3\times2L^2/3}$ denote the matrix whose elements are
given by
\begin{displaymath}
  \Mg_{m+L/3+L(n+L/2),p+L/3+L(q+L/2)}=\langle v_{n,m},v_{p,q}\rangle\,\;
  (n,p,m,q)\in S_{21}^2\times S_{22}^2\,.
\end{displaymath}
Since
\begin{equation}
\label{eq:51}
   \Big|\langle v_{n,m},v_{p,q}\rangle-\delta_{np}\delta_{mq}\Big|\leq
   Ce^{-\lambda L}\,,\;(n,p,m,q)\in S_{21}^2\times S_{22}^2 \,,
\end{equation}
where $\delta_{ij}$ denotes Kronecker delta, we may conclude that $\Mg$ is
diagonally dominant, and hence also that
$\Union_{n=-L/3+1}^{L/3}\Union_{m=-L/2+1}^{L/2}v_{n,m}$ forms a basis of
$\C^{2L^2/3}$ (a nearly orthonormal one).  
Furthermore, as
\begin{displaymath}
    \Big|\langle v_{-L/3,n},v_{p,q}\rangle-\frac{3}{2L^2}\sum_{k=1}^{\frac{2L^2}{3}} z_k^{-L/3}\,
    \delta_{p,L/3}\delta_{nq}\Big|\leq e^{-\lambda L}
    \,,\;(p,n,q)\in S_{21}\times S_{22}^2  \,, 
\end{displaymath}
we may conclude that 
\begin{equation}
\label{eq:52}
  \Big|v_{-\frac{L}{3},n}-\frac{3}{2L^2}\sum_{k=1}^{\frac{2L^2}{3}}
  z_k^{-L/3}v_{\frac{L}{3},n}\Big|\leq e^{-\lambda L}\,,\; \forall n\in S_{22}\,.
\end{equation}
It follows for the $k$'th component of \eqref{eq:52} that
\begin{displaymath}
    \Big|z_k^{-L/3}w_k^n-\frac{2}{3L^2}\sum_{j=1}^{\frac{2L^2}{3}}z_j^{-\frac{2L}{3}} z_k^{L/3}w_k^n\Big|\leq e^{-\lambda L}\,,
\end{displaymath}
and hence $z_k^{2L/3}$ is independent of $k$, up to an exponentially
small error. Consequently, there exists
$\alpha\in\R$ such that 
\begin{displaymath}
  \Big|e^{-i\alpha}-\frac{2}{3L^2}\sum_{j=1}^{\frac{2L^2}{3}}z_j^{-\frac{2L}{3}}\Big|\leq e^{-\lambda L} \,.
\end{displaymath}
from which we readily conclude that
\begin{equation}
\label{eq:53}
   1-e^{-\lambda L}\leq\Big|\frac{2}{3L^2}\sum_{j=1}^{\frac{2L^2}{3}}z_j^{-\frac{2L}{3}}\Big|
\end{equation}
Let $e^{i\alpha_j}=z_j^{-\frac{2L}{3}}$. Without any loss of generality we may assume that
\begin{displaymath}
  \max_{j\neq k}|\alpha_j-\alpha_k|=|\alpha_1-\alpha_2| \,.
\end{displaymath}
We now obtain the following upper bound
\begin{displaymath}
  \Big|\frac{2}{3L^2}\sum_{j=1}^{\frac{2L^2}{3}}z_j^{-\frac{2L}{3}}\Big|\leq
  \Big|\frac{2}{3L^2}\sum_{j=3}^{\frac{2L^2}{3}}z_j^{-\frac{2L}{3}}\Big|+\frac{2}{3L^2}|e^{i\alpha_1}+e^{i\alpha_2}|
  \leq  1-\frac{4}{3L^2} +\frac{2}{3L^2}|1+e^{i(\alpha_2-\alpha_1)}|\,,
\end{displaymath}
which together with the lower bound \eqref{eq:53} yields 
\begin{displaymath}
  2-e^{-\lambda L}\leq |1+e^{i(\alpha_2-\alpha_1)}|\,.
\end{displaymath}
It easily follows that there exists $\lambda_0>0$ such that for any $\lambda<\lambda_0$
it holds, for sufficiently large $L$, that
\begin{displaymath}
    \max_{j\neq k}|\alpha_j-\alpha_k|\leq e^{-\lambda L} \,.
\end{displaymath}
Consequently,
\begin{equation}
\label{eq:54}
  |z_k^{2L/3}-e^{i\alpha}|\leq e^{-\lambda L}\,,\;\forall k\in[1,2L^2/3]\cap\Z \,.
\end{equation}

To obtain $w_k^L$ we observe that 
\begin{equation}
\label{eq:55}
   \Big| \langle v_{n,-L/2},v_{p,q}\rangle-\frac{3}{2L^2}\sum_{k=1}^{\frac{2L^2}{3}}
    w_k^{-L}z_k^{\pm L/3}\delta_{q,L/2}\delta_{n,p\pm L/3}\Big|\leq e^{-\lambda L}
    \,,\;(n,p,q)\in S_{21}^2\times S_{22}  \,, 
\end{equation}
where the sign of $\pm$ is determined from the requirement $p\pm
L/3\in S_{21}$.  Suppose that $n\in[-L/3+1,0]\cap\Z$. Then we
obtain for the $k$'th component of \eqref{eq:55} that
\begin{displaymath}
   \Big|z_k^{n}w_k^{-L/2}- \frac{2}{3L^2}\sum_{j=1}^{\frac{2L^2}{3}}
   z_j^{-\frac{2L}{3}} z_k^{n+L/3}w_k^{L/2}\Big|\leq e^{-\lambda L}\,, 
\end{displaymath}
from which we conclude that there exists $\beta\in\R$ such that
\begin{displaymath}
  \Big|\frac{2}{3L^2}\sum_{j=1}^{\frac{2L^2}{3}}z_j^{\frac{L}{3}}w_j^L
  -e^{i\beta}\Big|\leq e^{-\lambda L}
\end{displaymath}
Hence, as in the proof of \eqref{eq:54} we obtain that there exists
$\lambda_0>0$ such that for all  $\lambda<\lambda_0$ it holds, for sufficiently large $L$
\begin{equation}
\label{eq:56}
   \big|w_k^L-e^{i\beta}z_k^{-L/3}\big|\leq e^{-\lambda L} \,,\;\forall k\in[1,2L^2/3]\cap\Z \,.
\end{equation}

Using the transformation 
\begin{equation}
\label{eq:57}
  \{z_k\}_{k=1}^{2L^2/3}
\to\{e^{i\alpha/L}z_k\}_{k=1}^{2L^2/3}\,,, 
\end{equation}
we may set $z_1=1$, yielding by \eqref{eq:54} that $|z_k^{2L/3}-1|\leq
e^{-\lambda L}$. Similarly we can set $w_1=1$, yielding $\beta=0$ in
\eqref{eq:56}. It follows that $|w_k^{2L}-1|\leq e^{-\lambda L}$. Since
\begin{displaymath}
 \Big|\frac{3}{2L^2}\sum_{k=1}^{2L^2/3} w_k^L\Big|\leq e^{-\lambda L} 
\end{displaymath}
we may conclude (by rearranging the indices) that
\begin{equation}
\label{eq:58}
  \big|w_k^L-(-1)^{k-1}\big| \leq e^{-\lambda L}\,, 
\end{equation}
Furthermore, by \eqref{eq:56} with $\beta=0$ and \eqref{eq:58}it holds that
\begin{displaymath}
  \big|z_k^{L/3}-(-1)^{k-1}\big| \leq e^{-\lambda L}\,, 
\end{displaymath}
which together with \eqref{eq:58} establishes \eqref{eq:49}.
\end{proof}
In the sequel may assume, without any loss of generality (by using  \eqref{eq:57}),  that
\eqref{eq:49} holds with $\alpha=\beta=0$.

Since the number of distinct solutions in $\C^2$ to the equations
  \begin{equation}
\label{eq:59}
  (z^{L/3},w^L)=(1,1) \quad \text{and} \quad   (z^{L/3},w^L)=(-1,-1) 
\end{equation}
is precisely $2L^2/3$ it follows from \eqref{eq:49} that either there
is exactly one point in $\{(z_k,w_k)\}_{k=1}^{\frac{2L^2}{3}}$ inside a
ball of radius $e^{-\lambda L}$ near each solution of \eqref{eq:59} or that
at least two points in $\{(z_k,w_k)\}_{k=1}^{\frac{2L^2}{3}}$ are
$\OO(e^{-\lambda L})$ distant from each other.  Consequntly, a proper lower
bound on $\min_{k_1\neq k_2}|(z_{k_1},w_{k_1})-(z_{k_2},w_{k_2})|$ will
allow us to obtain that there exists $\lambda_0>0$ such that for all
$0<\lambda<\lambda_0$ and any $m\in[1,2L^2/3]\cap\Z$ it holds that
\begin{equation}
  \label{eq:60}
(z_m,w_m)\in B\big((e^{3\pi ij/L},e^{2\pi il/L}e^{\frac{\pi
  i[1+(-1)^j]}{2L}}),e^{-\lambda L}\big)\,,
\end{equation}
for all $L>L_0(\lambda)$, where $j=\lfloor3m/2L\rfloor$, and $l=m-L\lfloor m/L\rfloor$. 

Such a lower bound is established in the following lemma.
\begin{lemma}
\label{lem:noduplicity}
  Let $\{(z_k,w_k)\}_{k=1}^{2L^2/3}$ denote a solution of
  \eqref{eq:48}. Let further \break $(k_1,k_2)\in([1,2L^2/3]\cap\Z)^2$ satisfy
  $k_1\neq k_2$. Then there exists positive $C$ and $L_0$ such that for
  all $L>L_0$ it holds that
  \begin{equation}
    \label{eq:61}
 \big|(z_{k_1},w_{k_1})-(z_{k_2},w_{k_2})\big|\geq\frac{C}{L}\,..
  \end{equation}
\end{lemma}
\begin{proof}
  The proof is based on the same arguments the
  Erd\H{o}s-Tur\'an-Koksma \cite[Theorem 5.2.1]{ha98} is based on . Thus,
  we set
  \begin{displaymath}
    (z_k,w_k)=(e^{2\pi ix_k},e^{2\pi iy_k})\,,
  \end{displaymath}
  where $(x_k,y_k)\in[0,1]^2$. Then, for some $a\in[0,1)$ and
  $0<\delta<1-a$, we let $\chi_{\delta,a}$ denote the characteristic function of
  $[a-\delta/2,a+\delta/2]$. We further set
\begin{displaymath}
  \chi^*_{\delta,a}(x)=\sum_{n\in\Z}\chi_{\delta,a}(x+n)
\end{displaymath}
It is well known  \cite[Section 5.2]{ha98} that for any $M\geq1$ there
exists an entire function $v_{\delta,a}^M$ satisfying
\begin{alignat*}{2}
  & \chi^*_{\delta,a}(x)\leq v_{\delta,a}^M(x)  & \forall x\in\R \\
  & v_{\delta,a}^M(x)=\sum_{n=-M}^M \hat{v}_n\exp \{2\pi inx\}
  & \forall x\in\R \\
  & \hat{v}_0=\delta+\frac{1}{M+1} \\
  &  |\hat{v}_n|\leq 2\Big(\delta +\frac{1}{M+1}\Big) & \forall1\leq|n|\leq M \,.
\end{alignat*}
Let $L_R$ denote the number of points $(x_n,y_n)$ associated with the
solution of \eqref{eq:48} in \linebreak
$R=[a-\delta/2,a+\delta/2]\times[b-\delta/2,b+\delta/2]$.  
As in \cite[Section 5.2]{ha98} we write
\begin{displaymath}
  L_R=\sum_{n=1}^{2L^2/3} \chi^*_{\delta,a}(x_n) \chi^*_{\delta,b}(y_n)\leq\sum_{n=1}^{2L^2/3} v_{\delta,a}^M(x_n) v_{\delta,b}^M(y_n) 
\end{displaymath}
It follows that
\begin{displaymath}
  L_R\leq \Big[\delta L+\frac{L}{M_w+1}\Big]\Big[\delta L+\frac{2L}{3(M_z+1)}\Big]\Big[1+2\sum_{
    \begin{subarray}\strut
      m=-M_z \\
      m\neq 0
    \end{subarray}}^{M_z}
  \Big|\sum_{k=1}^{2L^2/3}z_k^m \Big|\Big] \Big[1+2\sum_{
    \begin{subarray}\strut
      n=-M_w \\
      n\neq0
    \end{subarray}
  }^{M_w} \Big|\sum_{k=1}^{2L^2/3}w_k^n\Big|\Big]\,.
\end{displaymath}
Choosing $M_z=2L/3-1$ and $M_w=L-1$ yields, by \eqref{eq:48}
\begin{displaymath}
   L_R\leq \Big[\delta L+1\Big]^2(1+e^{-\lambda L})\,.
\end{displaymath}
For $\delta L<\sqrt{2}-1$ we obtain that $L_R<2$ for sufficiently large
$L$ proving therefore
that for $k_1\neq k_2$ it holds that 
\begin{displaymath}
  |(x_{k_1},y_{k_1})-(x_{k_2}y_{k_2})|\geq \frac{\sqrt{2}-1}{2L}\,.
\end{displaymath}
From the above, we may easily conclude \eqref{eq:61} 
and hence also \eqref{eq:60}\,.
\end{proof}

We continue as in Section \ref{sec:2}. 
For convenience, we use translation to obtain that
\begin{equation}
\label{eq:62}
\frac{\sqrt{3}}{2L^2}\sum_{n=1}^{2L^2/3} x_n=\frac{3}{2L^2}\sum_{n=1}^{2L^2/3} y_n=\frac{1}{2}\,,
\end{equation}
and set
\begin{subequations}
\label{eq:63}
  \begin{equation}
  x_{kj}^0=\frac{3\sqrt{3}(2k-1)}{4L}\quad ; \quad y_{kj}^0=\frac{4j+[1+(-1)^k]}{4L}
\end{equation}
and
  \begin{equation}
    \Tg_L= \Union_{k=1}^{2L/3}\Union_{j=1}^L( x_{kj}^0, y_{kj}^0)\,.
  \end{equation}
\end{subequations}
Let $\Sg_L=\Union_{k=1}^{2L/3}\Union_{j=1}^L\{(x_{kj},y_{kj})\}$ denote the minimizer of $\Jg_L$ satisfying
\eqref{eq:62}. We now set ${\bs \delta}_L=\Sg_L-\Tg_L$. Set further
\begin{equation}
  \label{eq:64}
{\bs \delta}_L=({\bs \delta}_L^x,{\bs \delta}_L^y)\,,
\end{equation}
where ${\bs \delta}_L^x\in\C^{2L^2/3}$ represents the difference in $x$ between $\Sg_L$
and $\Tg_L$ and ${\bs \delta}_N^y\in\C^{2L^2/3}$ similarly represents the difference in
$y$.

Let now  ${\bs
  V}_{mn}\in\C^{2L^2/3}$ be given, for $|m|\geq1$ by
\begin{displaymath}
 [{\bs V}_{mn}]_{(k-1)L+j} =
 \frac{\sqrt{3}}{2\pi im}\frac{\partial b_{mn}}{\partial x_{kj}}(\Tg_L)\,,\;(k,j)\in I_L \,.
\end{displaymath}
where
\begin{displaymath}
  I_L=[1,2L/3]\cap\Z\times[1,L]\cap\Z\,.
\end{displaymath}
Hence, 
\begin{equation}
\label{eq:65}
   [{\bs V}_{nm}]_{(k-1)L+j} =\frac{3}{2L^2}\exp
   \Big\{i\pi\frac{6mk+4nj+n[1+(-1)^k]}{2L}\Big\}\,.
 \end{equation}
Note that for $|n|\geq 1$
\begin{displaymath}
 [{\bs V}_{mn}]_{(k-1)L+j} =\frac{1}{2\pi in}
 \frac{\partial b_{mn}}{\partial y_{kj}}(\Tg_L)\,,\; (k,j)\in I_L \,.
\end{displaymath} 

We can now state
\begin{lemma}
 There exists $C>0$, independent of $L$, such that
 \begin{equation}
\label{eq:66}
\sum_{m=1}^{2L/3}\sum_{n=-L+1}^L|b_{mn}|^2 \geq CL^2\|{\bs \delta}_L\|_2^2 \,.
 \end{equation}
\end{lemma}
\begin{proof}

{\em Step 1:}   Show that there exists $C>0$ such that for
sufficiently large $L$ it holds that
  \begin{equation}
\label{eq:67}
\sum_{m=1}^{2L/3}\sum_{n=-L+1}^L|\langle{\bs V}_{mn},{\bs \delta}_{mn}\rangle|^2\geq CL^2 \|{\bs \delta}_L\|_2^2 \,,
  \end{equation}
where 
\begin{displaymath}
  {\bs \delta}_{mn} =m{\bs \delta}_L^x+n{\bs \delta}_L^y\,.
\end{displaymath}

We first note that
$\Vg_L=\Union_{(m,n)\in I_L}{\bs V}_{mn}$ 
forms an orthonormal basis of $\C^{2L^2/3}$.
To prove the above statement we compute
\begin{equation}
\label{eq:68}
  \langle{\bs V}_{mn},{\bs V}_{pq}\rangle =\frac{9}{4L^4}\sum_{j=1}^{2L/3}\sum_{l=1}^L\exp
   \Big\{i\pi\frac{6(m-p)j+(n-q)(4l+[1+(-1)^j])}{2L}\Big\}=\delta_{mp}\delta_{nq}\,.
\end{equation}
Hence,
\begin{displaymath}
  {\bs \delta}_L^x=\sum_{(m,n)\in I_L}\alpha_x^{mn}{\bs V}_{mn} \quad;\quad {\bs
    \delta}_L^y=\sum_{(m,n)\in I_L}\alpha_y^{mn}{\bs V}_{mn}\,, 
\end{displaymath}
where
\begin{displaymath}
  \alpha_x^{mn}=\langle {\bs V}_{mn},{\bs \delta}_L^x\rangle \quad ;\quad   \alpha_y^{mn}=\langle {\bs V}_{mn},{\bs \delta}_L^y\rangle \,.
\end{displaymath}
Note that since ${\bs \delta}_L$ is real we must have
\begin{displaymath}
  \alpha_x^{mn}= \alpha_x^{(2L/3-m),(L-n)}\quad ;\quad \alpha_y^{mn}= \alpha_y^{(2L/3-m),(L-n)}\,.
\end{displaymath}
Hence,
\begin{subequations}
\label{eq:69}
  \begin{equation}
  \sum_{(m,n)\in I_L}|\langle{\bs V}_{mn},{\bs \delta}_{mn}\rangle|^2=\frac{1}{2}
\sum_{m=1}^{2L/3}\sum_{n=1}^L\Big[|m\alpha_x^{mn}+n\alpha_y^{mn}|^2+
\Big|\Big(\frac{2L}{3}-m\Big)\alpha_x^{mn}+(L-n)\alpha_y^{mn}\Big|^2\Big]
\end{equation}
Similarly, we obtain that
\begin{equation}
  \sum_{(m,-n+1)\in I_L}|\langle{\bs V}_{mn},{\bs \delta}_{mn}\rangle|^2=\frac{1}{2}
\sum_{m=1}^{2L/3}\sum_{n=1}^L\Big[|m\alpha_x^{mn}+(n-L)\alpha_y^{mn}|^2+
\Big|\Big(\frac{2L}{3}-m\Big)\alpha_x^{mn}-n\alpha_y^{mn}\Big|^2\Big]
\end{equation}
\end{subequations}

We now write
\begin{subequations}
\label{eq:70}
  \begin{multline}
  |m\alpha_x^{mn}+n\alpha_y^{mn}|^2+
\Big|\Big(\frac{2L}{3}-m\Big)\alpha_x^{mn}+(L-n)\alpha_y^{mn}\Big|^2+|m\alpha_x^{mn}+(n-L)\alpha_y^{mn}|^2+
\\
\Big|\Big(\frac{2L}{3}-m\Big)\alpha_x^{mn}-n\alpha_y^{mn}\Big|^2=a(\alpha_x^{mn})^2+2b\alpha_x^{mn}\alpha_y^{mn}+c(\alpha_y^{mn})^2\,,
\end{multline}
where
\begin{align}
  & a=4\Big(\frac{L}{3}-m\Big)^2+\frac{4L^2}{9} \\
  & b=4\Big(\frac{L}{3}-m\Big)\Big(\frac{L}{2}-n\Big) \\
  & c=4\Big(\frac{L}{2}-n\Big)^2+L^2\,.
\end{align}
\end{subequations}
Given that
\begin{displaymath}
  ac-b^2\geq \frac{26}{9}L^2 \,,
\end{displaymath}
we can now conclude that there exist $C>0$ such that 
\begin{displaymath}
  a(\alpha_x^{mn})^2+2b\alpha_x^{mn}\alpha_y^{mn}+c(\alpha_y^{mn})^2\geq CL^2\big[(\alpha_x^{mn})^2+(\alpha_y^{mn})^2\big]\,.
\end{displaymath}
Combining the above with \eqref{eq:69} and \eqref{eq:70} yields
\eqref{eq:67}.

{\em Step 2:} Prove that there exist $\lambda>0$ and $L_0>0$ such that
\begin{equation}
\label{eq:71}
 \sum_{m=1}^{2L/3}\sum_{n=-L+1}^L|b_{mn}-2\pi i{\bs
   V}_{mn} \cdot{\bs \delta}_{mn}  |^2 \leq Ce^{-\lambda L}\|{\bs
   \delta}_{mn}\|_2^2  \,,
\end{equation}
for all $L>L_0$.

Writing the Taylor expansion of $b_{mn}$, for $(m,n)\in I_L$, around
$\Tg_L$ gives
\begin{displaymath}
b_{mn} =   2\pi i{\bs V}_{mn} \cdot {\bs \delta}_{mn} -
\frac{3\pi^2}{L^2}\sum_{k=1}^{2L/3}\sum_{j=1}^Le^{i2\pi\tilde{z}_{kj}^{mn}}|\delta_{mn}^{kj}|^2 \,,
\end{displaymath}
where
\begin{displaymath}
  \tilde{z}_{kj}^{mn}= mx_{kj}^0+ny_{kj}^0+ \theta\big(m[x_{kj}-x_{kj}^0]+ n[y_{kj}-y_{kj}^0]\big)\,,
\end{displaymath}
and 
\begin{displaymath}
  \delta_{mn}^{kj}=[{\bs \delta}_{mn}]_{(k-1)L+j} 
\end{displaymath}
Consequently,
\begin{displaymath}
  |b_{nm}-2\pi i{\bs V}_{mn} \cdot {\bs \delta}_{mn} |\leq \frac{3\pi^2}{L^2}\|\delta_{mn}\|_2^2
\end{displaymath}
Hence,
\begin{displaymath}
 \sum_{m=1}^{2L/3}\sum_{n=-L+1}^L|b_{nm}-2\pi i{\bs V}_{mn} \cdot {\bs \delta}_{mn} |^2
 \leq \frac{9\pi^4}{L^4} \sum_{m=1}^{2L/3}\sum_{n=-L+1}^L\|\delta_{mn}\|_2^4\,.
\end{displaymath}
By \eqref{eq:60}we then obtain \eqref{eq:71}. 

{\em Step 3:} Prove \eqref{eq:66}.

As
\begin{displaymath}
  \sum_{m=1}^{2L/3}\sum_{n=-L+1}^L|b_{mn}|^2  \geq \frac{1}{2}\sum_{m=1}^{2L/3}\sum_{n=-L+1}^L|\langle{\bs V}_{mn},{\bs \delta}_{mn}\rangle|^2
  -
 \sum_{m=1}^{2L/3}\sum_{n=-L+1}^L|b_{mn}-2\pi i{\bs
   V}_{mn} \cdot{\bs \delta}_{mn}  |^2 \,,
\end{displaymath}
we obtain from \eqref{eq:67}, \eqref{eq:71} and \eqref{eq:60} that
\begin{displaymath}
  \sum_{m=1}^{2L/3}\sum_{n=-L+1}^L|b_{mn}|^2\geq C[L^2 \|{\bs \delta}_L\|_2^2-e^{-\lambda L}\|{\bs \delta}_L\|_2^2 ]\,.
\end{displaymath}
For sufficiently large $L$ we can easily conclude \eqref{eq:66}.
\end{proof}

We now establish that a minimizing set under \eqref{eq:15} is the
triangular lattice

\begin{proof}
  Let $b_n=b_n(\Sg_L)$. We first observe that
  \begin{multline}
\label{eq:72}
    \Jg(\Sg_L)-\Jg(\Tg_L)\geq \sum_{m=1}^{\frac{2L}{3}}\sum_{n=-L+1}^L
    a_{mn}|b_{mn}|^2\\ - \frac{1}{2}\sum_{(m,n)\in\Z^2\setminus\{(0,0)\}} a_{mL,nL}[1+(-1)^{m+n}][1-|b_{mL,nL}|^2] \,.
  \end{multline}
We first obtain a lower bound for the first term on the
right-hand-side using \eqref{eq:66}and  \eqref{eq:17}
\begin{equation}
\label{eq:73}
\sum_{m=1}^{\frac{2L}{3}}\sum_{n=-L+1}^L
    a_{mn}|b_{mn}|^2 \geq
    Ca_{2L/3,L}\sum_{m=1}^{2L/3}\sum_{n=-L+1}^L|b_{mn}|^2 \geq
    CL^2a_{2L/3,L}\|{\bs \delta}_L\|_2^2  \,.
\end{equation}

We next obtain an upper bound for the second term on the
right-hand-side of \eqref{eq:73}. In a similar manner to the one used
to obtain \eqref{eq:37} we obtain that
\begin{multline*}
\sum_{(m,n)\in\Z^2\setminus\{(0,0)\}} a_{mL,nL}[1+(-1)^{m+n}][1-|b_{mL,nL}|^2] \leq\\
CL^4\|{\bs \delta}_L\|_2^2\sum_{(m,n)\in\Z^2\setminus\{(0,0)\}} (m^2+n^2)a_{mL,nL}[1+(-1)^{m+n}] \,. 
\end{multline*}
By \eqref{eq:17} it holds that for sufficiently large $L$
\begin{equation}
\label{eq:74}
  \sum_{(m,n)\in\Z^2\setminus\{(0,0)\}} a_{mL,nL}[1+(-1)^{m+n}][1-|b_{mL,nL}|^2]
  \leq CL^4e^{-\lambda L}a_{2L/3,L}\|{\bs \delta}_L\|_2^2
\end{equation}
Substituting \eqref{eq:73} and \eqref{eq:74} into \eqref{eq:72} yields
that for sufficiently large $L$ 
\begin{displaymath}
  \Jg(\Sg_L)-\Jg(\Tg_L)\geq0 \,,
\end{displaymath}
where equality is achieved only when $\Sg_L=\Tg_L$.  The Theorem
is proved.
\end{proof}

\begin{remark}
  Let $G:\R^2\to\R$ denote the kernel of the periodic inverse of $-\Delta$
  on \linebreak $R=[0,1]\times[0,1]$. Then (see \cite{pese20} for instance)
  \begin{equation}
\label{eq:75}
    G({\bs x},0)= \int_0^\infty S_t(\bs x,0)\, dt \,,
  \end{equation}
where the periodic heat kernel is given by
\begin{displaymath}
  S_t(\bs x,0)= \sum_{\bs \xi\in\Z^2}\frac{1}{4\pi t}\exp\Big\{\frac{|x-\xi|^2}{t}\Big\}  -1\,.
\end{displaymath}
By the Poisson summation formula we have
\begin{displaymath}
 S_t(\bs x,0) = 
\sum_{\bs \omega\in\Z^2 \setminus \{0\}} e^{-4\pi^2 |\omega|^2 t } e^{2i\pi \bs \omega\cdot \bs x}.
\end{displaymath}
Thus, stretching one of the coordinates by $\sqrt{3}$, we obtain from
the results of this section that the triangular is the minimizer, with
respect to triplets, for all $t\gg1/L^2$. However, since the integration
with respect $t$ in \eqref{eq:75} starts from $t=0$ we cannot yet
conclude that the minimizer of $G$ with respect to triplets is the
triangular lattice. 
\end{remark}

{\bf Acknowledgments:} The author would like to thank Itai Shafrir for
some helpful discussions.
\bibliography{vortices1}
\end{document}